\documentclass[twocolumn,superscriptaddress,10pt]{revtex4-2}
\usepackage{multirow,amsthm,amssymb,amsbsy,amsmath,bm,float}
\usepackage{hyperref}
\usepackage{graphicx}
\usepackage{booktabs}
\usepackage{verbatim}
\usepackage{dcolumn}
\usepackage{color}
\usepackage{xcolor}
\usepackage[normalem]{ulem}
\usepackage{soul}
\usepackage{braket}
\usepackage{physics}

\theoremstyle{plain}
\newtheorem{theorem}{Theorem}
\newtheorem{lemma}{Lemma}
\newtheorem{corollary}{Corollary}

\theoremstyle{definition}

\begin{document}

\title{A hidden bottleneck in  classical and quantum linear reservoir computing}

\author{Johannes Nokkala}
\affiliation{
 Department of Physics and Astronomy, University of Turku, FI-20014, Turun Yliopisto, Finland
}
\affiliation{
 Turku Collegium for Science, Medicine and Technology, University of Turku, Turku, Finland.
}

\author{Federico Centrone}
\affiliation{
ICFO-Institut de Ciencies Fotoniques, The Barcelona Institute of Science and Technology,
Av. Carl Friedrich Gauss 3, 08860 Castelldefels (Barcelona), Spain
}
\affiliation{
Universidad de Buenos Aires, Facultad de Ciencias Exactas y Naturales,
Departamento de F{\'i}sica, Ciudad Universitaria, 1428 Buenos Aires, Argentina
}

\author{Francesco Arzani}
\affiliation{QAT team, DIENS, \'Ecole Normale Sup\'erieure, PSL University, CNRS, INRIA, 45 rue d'Ulm, Paris 75005, France
}

\begin{abstract}
We identify a hidden bottleneck in the information processing capacity of linear reservoir computers. When the measured features evolve linearly in the reservoir and the output is formed by linear readout with bias, we show that the capacity available at any fixed delay is limited by what is already present in the preprocessed input. Linear reservoir dynamics can therefore redistribute features, but cannot create new fixed-delay expressive power on their own. This limitation is hidden by global capacity measures, since contributions from different delays can accumulate even when each individual delay is strongly constrained. As an experimentally important realization of this general result, we derive the corresponding Gaussian limit for covariance-based continuous-variable quantum reservoirs. Numerical experiments show that experimentally accessible single-photon operations surpass this limit, establishing them as a genuine resource for quantum reservoir computing. The resulting excess capacity also provides an operational witness of non-Gaussian processing in black-box continuous-variable systems under minimal assumptions.
\end{abstract}

\maketitle
\section{Introduction}

Nonlinearity is a fundamental and ubiquitous aspect of the natural world, governing phenomena from fluid dynamics to the complex interactions within biological systems. For computation, and particularly for machine learning, this nonlinearity is not a nuisance but an essential ingredient for powerful information processing. By mapping inputs into higher-dimensional feature spaces, nonlinear transformations allow computational models to capture intricate patterns and correlations that are inaccessible to purely linear methods. This principle underpins the remarkable success of modern neural networks. However, the power of nonlinearity is a double-edged sword. Systems with strong nonlinear feedback are notoriously difficult to analyze, their dynamics can be numerically unstable, and proving formal properties about their behavior is often intractable. A prominent example is the difficulty of training recurrent neural networks, where feedback loops make gradient-based optimization prone to instability~\cite{pmlr-v28-pascanu13}.

Reservoir computing offers an elegant solution to this dilemma~\cite{10.1016/j.cosrev.2009.03.005}. Derived from recurrent neural network theory, it circumvents the challenges of nonlinear training by adopting a radical separation of concerns. 
 The core of the system, known as the  ``reservoir'', is a fixed, nonlinear dynamical system that is never trained. 
 Instead, it acts as a high-dimensional feature generator, activated by an input signal. The computational task is learned only at the final stage by a linear readout layer, which is simple and efficient to train. This paradigm shows that rich, recurrent dynamics are possible and powerful, provided the nonlinear part remains fixed and only the linear part is trained.

This approach is particularly well-suited for processing time-series data~\cite{10.1016/j.cosrev.2009.03.005,TANAKA2019100}, where its intrinsic memory and complex transient responses can be effectively exploited. 
 Furthermore, its appeal for experimental science is immense. Because the framework is largely agnostic about the specific physical realization of the reservoir, requiring only a suitable set of nonlinear dynamics, it has emerged as a versatile platform for computation in a vast range of physical systems. This flexibility is especially attractive for quantum platforms, where engineering controllable, nonlinear quantum dynamics is a central goal.

At the same time, the use of a linear readout raises a basic conceptual question: where does the computational power of a reservoir actually come from? More specifically, how much of the nonlinearity is produced intrinsically by the reservoir, and how much is already supplied by input encoding or classical post-processing? This question is particularly pressing in architectures where feature generation and memory are not provided by the same mechanism. In many practically relevant schemes, nonlinear transformations are already introduced at the input stage, while memory is carried by a separate dynamical or feedback layer; in quantum reservoir computing (QRC), nonlinear input encoding has been identified as a widespread feature of existing proposals, and in continuous-variable Gaussian reservoirs the encoding itself is a primary source of nonlinearity \cite{govia2022nonlinear,nokkala2021gaussian}. It is therefore natural to ask whether this architectural separation imposes a fundamental limitation on information processing capacity.

In this work we show that it does. We identify a hidden bottleneck in the information processing capacity (IPC) of linear reservoirs. For reservoir computers whose chosen observables admit a finite-dimensional linear update and whose output is formed by linear readout with bias, the capacity available at any fixed delay is bounded by the capacity already present in the preprocessed input. In other words, linear reservoir dynamics cannot create new delay-resolved expressive power on their own; they can only linearly reshape what preprocessing has already made available. This is a strong limitation, but it is hidden at the level of aggregate capacity because contributions from different delays can accumulate. A reservoir may therefore appear expressive in total IPC, or on benchmark tasks, even when the expressive power available at each individual delay is sharply constrained. Our result provides a crucial design principle for future reservoir computers, rigorously defining the boundary between what can be achieved with simple linear dynamics and what requires a genuine investment in engineering nonlinear physical interactions.

This general viewpoint is especially relevant for continuous-variable (CV) photonic quantum reservoir computing, where input encoding is a key source of nonlinearity and Gaussian dynamics often mediate the reservoir evolution \cite{nokkala2021gaussian,nokkala2021high,garcia2023scalable,garcia2023squeezing,nokkala2023online,nokkala2024retrieving,paparelle2025experimental}. CV photonic platforms based on Gaussian states, Gaussian transformations, and homodyne detection are among the most experimentally promising realizations of QRC. They are scalable, resilient to decoherence, compatible with room-temperature operation, and their use for QRC has recently been experimentally demonstrated \cite{paparelle2025experimental}. However, Gaussian dynamics act linearly on position and momentum variables, which makes them efficiently classically simulable \cite{bartlett2002efficient}, so  their computational power is expected to be structurally limited. Our general bottleneck theorem makes this intuition precise: Gaussian covariance dynamics provide an experimentally important realization of a broader linear-reservoir limitation.

For covariance-based CV quantum reservoirs, the general linear-reservoir bottleneck yields a concrete Gaussian bound. At any fixed delay, the IPC of a Gaussian reservoir is bounded by the IPC already contained in the encoded input covariance matrix. This immediately clarifies the role of non-Gaussian resources. If the observed delay-resolved capacity exceeds the Gaussian bound, then no Gaussian reservoir transformation compatible with the stated assumptions can explain the data. The resulting excess capacity is therefore an operational witness of non-Gaussian processing at the level of the effective reservoir map. This witness is extracted directly from input-output processing data and does not require full state tomography.

To demonstrate how the Gaussian bottleneck can be surpassed in practice, we introduce a minimal and experimentally realistic non-Gaussian CV QRC scheme. The input time series modulates squeezed-vacuum states, correlations are generated by a fixed Gaussian multimode transformation, and non-Gaussianity is injected through conditional single-photon operations, which have already been demonstrated experimentally in multimode settings \cite{roeland2022mode,wenger_non-gaussian_2004,zavatta_quantum--classical_2004,ourjoumtsev_generating_2006,averchenko_multimode_2016,ra_non-gaussian_2020}. The readout is formed from homodyne statistics, so the scheme remains close in spirit to existing Gaussian photonic reservoirs while introducing the minimal non-Gaussian ingredient needed to surpass the bound.

Our analysis combines rigorous analytical bounds with numerically exact evaluation of weakly non-Gaussian multimode states \cite{stornati2024variational}. We quantify nonlinear processing through IPC \cite{dambre2012information}, and model finite-ensemble effects through sampling distributions that are exact for Gaussian states and asymptotically accurate for the weakly non-Gaussian states considered here \cite{wishart1928generalised}. This allows us to assess not only ideal expressive power but also the operational visibility of the non-Gaussian advantage in realistic finite-resource regimes.

The main contributions of this work are threefold. First, we derive a general delay-resolved IPC bottleneck for linear reservoirs and identify covariance-based Gaussian CV quantum reservoirs as a central experimentally relevant realization. Second, we show numerically that experimentally accessible single-photon operations break the corresponding Gaussian bound, thereby establishing them as a genuine computational resource for CV QRC. Third, we show that the resulting excess capacity can be interpreted as an operational witness of non-Gaussian processing under minimal assumptions.

The rest of this paper is organized as follows. Section~II introduces the IPC framework and establishes the linear-reservoir bottleneck. Section~III presents the continuous-variable setting and the non-Gaussian photonic reservoir scheme considered in this work. Section~IV demonstrates the breaking of the Gaussian bound, the role of memory, and the finite-resource visibility of the non-Gaussian advantage. Section~V discusses the implications of these results, including the witness interpretation of excess capacity and the broader design principles suggested by the linear-reservoir bottleneck. Technical proofs and computational details are collected in the appendices.

\section{Information processing capacity of linear reservoirs}

\subsection{Delay-resolved information processing capacity}

\begin{figure*}[t]
\centering
\includegraphics[trim=0cm 4.5cm 0cm 0cm,clip=true,width=0.95\textwidth]{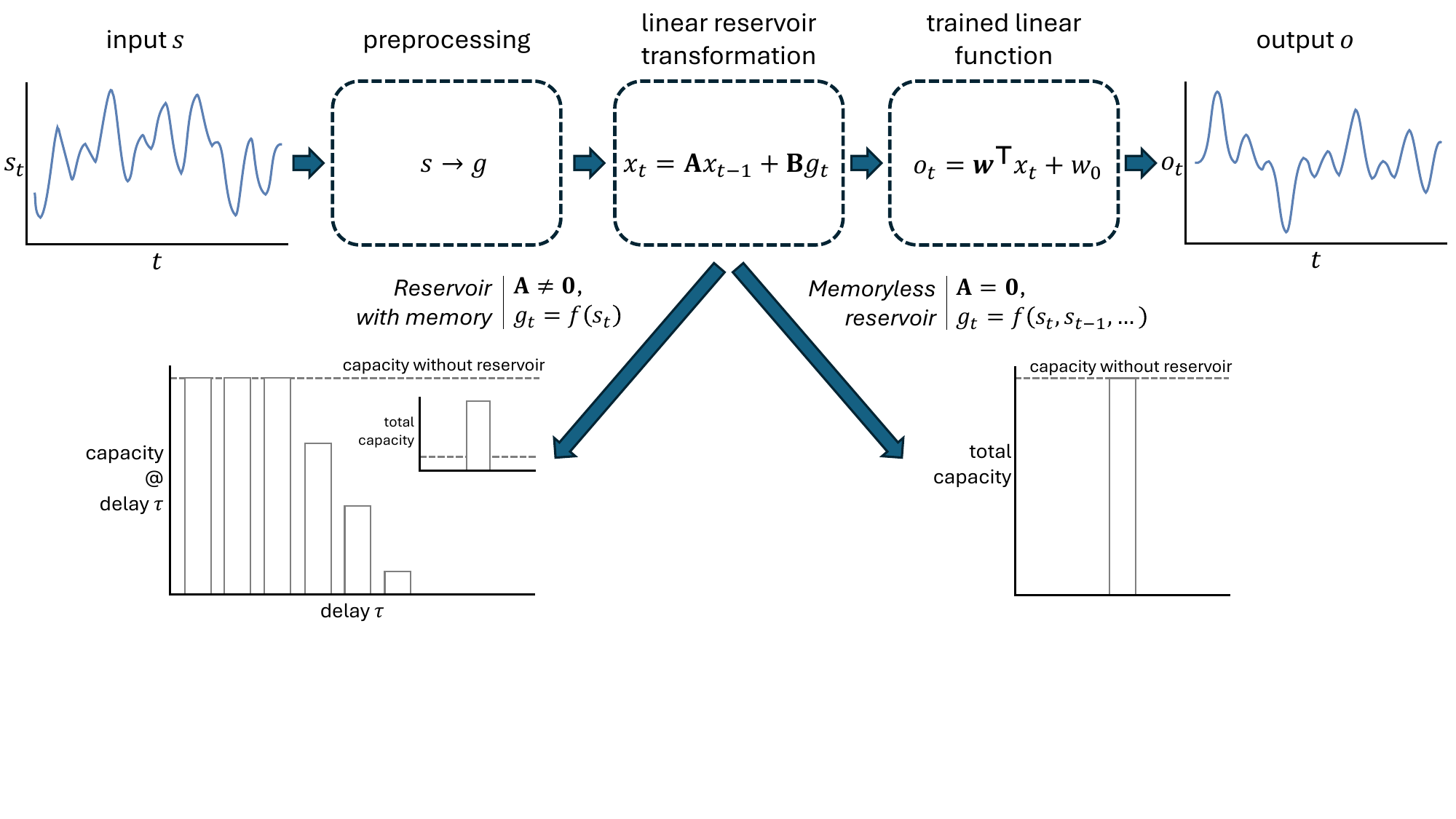}
\caption{\label{fig:limit} {\bf Schematic illustration of the linear-reservoir bottleneck.} Information processing capacity (IPC) is a task-independent measure of a dynamical system’s ability to retain and transform input information. In linear reservoirs, preprocessing may already generate nonlinear features, but the reservoir transformation itself cannot create new expressive power beyond what is present in those features. If memory is carried by the reservoir, the IPC at each fixed delay is bounded by the preprocessing-only IPC, even though capacities from different delays may accumulate in the total IPC. If memory is carried by preprocessing, the total IPC cannot exceed the preprocessing-only IPC. }
\end{figure*}

Reservoir computing is concerned with approximating temporal input-output maps. Let
\begin{equation}
s=\{\ldots,s_{t-1},s_t,s_{t+1},\ldots\}
\end{equation}
be an input time series, where each $s_t\in\mathbb{R}$ is indexed by the discrete time step $t$. The relevant targets are fading-memory functions, namely functions whose value at time $t$ depends continuously on a finite, or effectively finite, portion of the recent input history. In practice, the reservoir is driven by $s$, a set of observables is recorded, and only a final linear combination of those observables is trained. Further details are given in Appendix~\ref{app:training}.

To quantify the expressive power of a reservoir independently of any single task, we use the information processing capacity (IPC) \cite{dambre2012information}. For a target function $f$, the capacity is
\begin{equation}
C_f = 1-\mathrm{NMSE}_f,
\end{equation}
where $\mathrm{NMSE}_f$ is the normalized mean squared error achieved by the optimal linear readout. The IPC is obtained by summing these capacities over an orthonormal basis of fading-memory functions.

For concreteness, and following the standard construction of Ref.~\cite{dambre2012information}, we take the input to be independent identically distributed (iid) as $s_t\in[-1,1], \ s_t\sim \mathrm{Uniform(-1,1)}$ and use Legendre basis functions to resolve the dependence on nonlinearity degree and delay. A basis element of degree $d$ and delay $\tau$ is
\begin{equation}
P_d^\tau(s_t):=P_d(s_{t-\tau}),
\end{equation}
where $P_d$ denotes the Legendre polynomial of degree $d$, with the conventional normalization understood so that the resulting basis is orthonormal in the Hilbert space of fading memory functions. This makes it possible to resolve capacity not only by nonlinearity degree, but also by delay. We therefore define the delay-resolved IPC
\begin{equation}
\mathrm{IPC}^{(\tau)} := \sum_{d\ge 0} C_{P_d^\tau}.
\end{equation}

This quantity will play a central role below. While the total IPC measures the overall expressive power of a reservoir, $\mathrm{IPC}^{(\tau)}$ reveals how much expressive power is available at a specific delay. A reservoir may therefore have a large total IPC even if its capacity at each individual delay is strongly constrained. The bottleneck identified in this work is precisely of this type.

Recent work has made clear that information-processing capacity is not merely a task-independent diagnostic, but is directly connected to task performance: the error on a given task can be expressed as a weighted combination of the relevant IPC components, and total IPC alone may correlate poorly with task-specific accuracy \cite{hulser2023deriving}. This makes the delay-resolved viewpoint particularly useful. For example, a fixed-delay nonlinear memory task such as reconstructing $P_d(s_{t-\tau})$ probes a single delay sector and is therefore directly constrained by the bottleneck studied below \cite{dambre2012information}. By contrast, tasks such as delayed XOR or parity involve nonlinear combinations of information drawn from multiple past inputs and therefore depend on several delay sectors at once \cite{wringe2025reservoir,borghi2021reservoir}. The distinction between these two task classes will be important when interpreting the consequences of the linear-reservoir bottleneck.

\subsection{Linear reservoirs}

We consider reservoir computers in which the relevant observables can be collected into a finite-dimensional state vector $x_t\in\mathbb{R}^m$, and the preprocessed input into a vector $g_t\in\mathbb{R}^p$. The reservoir output is formed by a trained linear readout
\begin{equation}
o_t = \mathbf{w}^\top x_t + w_0,
\end{equation}
where $\mathbf{w}$ is a vector of trained weights and $w_0$ is a trained bias term.

The key assumption is that, in these chosen observables, the reservoir admits a finite-dimensional linear update,
\begin{equation}
x_t = \mathbf{A}x_{t-1}+\mathbf{B}g_t,
\label{eq:linear_reservoir_update}
\end{equation}
where $\mathbf{A}$ and $\mathbf{B}$ are constant matrices of appropriate dimension. More general affine updates can be reduced to this form by augmenting the state with the constant direction already accounted for by the bias term; for clarity, this standard step is given in the Appendix~\ref{app:affine}.

Equation~\eqref{eq:linear_reservoir_update} covers two qualitatively distinct situations:
\begin{enumerate}
    \item \textit{Memory in preprocessing, memoryless reservoir.} In this case $\mathbf{A}=\mathbf{0}$, whereas $g_t$ may already depend on several past inputs.
    \item \textit{Memoryless preprocessing, linear reservoir memory.} In this case $g_t=g(s_t)$ depends only on the current input, while temporal memory is generated by the reservoir recurrence through $\mathbf{A}$.
\end{enumerate}
We assume throughout that the relevant fading-memory conditions hold, so that the asymptotic IPC statements are well defined. 

Note that the scenario where memory is available both in the preprocessing and reservoir is also interesting: $\boldsymbol{A}\neq 0$, $g_t = g(s_t,s_{t-1},\ldots)$. However, the dynamics becomes much more complicated, as this would correspond to a deep (multi-layer) reservoir computer, and we leave this analysis for future work.

\subsection{Main theorem}

The key result is that linear reservoir transformations cannot create new delay-resolved expressive power. They can only redistribute or linearly reshape what has already been generated by preprocessing. In practical terms, the IPC of the preprocessing alone corresponds to using $o_t=\mathbf{w}^\top g_t+w_0$ instead of $o_t=\mathbf{w}^\top x_t+w_0$.

\begin{theorem}[]
\label{thm:linear_bottleneck}
Consider a reservoir computer of the form in Eq.~\eqref{eq:linear_reservoir_update} with linear readout and bias.

\begin{enumerate}
    \item[\textnormal{(A)}] If preprocessing has memory but the reservoir transformation is memoryless, then the total IPC after the reservoir is upper bounded by the IPC of the preprocessing features alone.
    
    \item[\textnormal{(B)}] If preprocessing is memoryless but the reservoir has linear memory, then for every fixed delay $\tau$ the delay-resolved IPC after the reservoir is upper bounded by the IPC of the preprocessing features alone.
\end{enumerate}
\end{theorem}

The proof is simple in structure. In case \textnormal{(A)}, the reservoir acts only as a linear map on the features produced by preprocessing and therefore cannot increase the number of linearly independent fading-memory functions. In case \textnormal{(B)}, iterating Eq.~\eqref{eq:linear_reservoir_update} expresses each observable as a sum of terms depending on one delay at a time. Basis functions supported on different delays are orthogonal, so when reconstructing a target at delay $\tau$, only the term with the matching delay can contribute; all other delays are irrelevant. The corresponding delay-resolved capacity is therefore bounded by what preprocessing alone already provides. The detailed proof is given in the Appendix~\ref{app:proof_linear}. The theorem is illustrated in Fig.~\ref{fig:limit}.

The theorem identifies a hidden bottleneck. Since capacities from different delays can accumulate, the total IPC may remain large even when the expressive power available at any fixed delay is sharply limited.

Theorem~\ref{thm:linear_bottleneck} applies to any reservoir architecture whose chosen observables admit a finite-dimensional linear update. In the remainder of this paper, we focus on an experimentally important realization of this general setting, namely continuous-variable quantum reservoirs with covariance-based readout. In that case, Gaussian dynamics provide a concrete instance of the linear-reservoir bottleneck, while single-photon operations offer a minimal route beyond it.

\section{Continuous-variable quantum reservoirs}

We now turn to the experimentally relevant realization considered in the remainder of this work, namely continuous-variable (CV) quantum reservoirs implemented with bosonic optical modes. Our goal in this section is not to review CV quantum optics in full, but only to introduce the minimal notions needed to understand why Gaussian reservoirs fall within the scope of Theorem~\ref{thm:linear_bottleneck}, and how a minimal non-Gaussian extension can break the corresponding bound.

\subsection{Bosonic modes, quadratures, and covariance matrices}

A bosonic mode is described by annihilation and creation operators $\hat a_j$ and $\hat a_j^\dagger$ satisfying
\begin{equation}
[\hat a_j,\hat a_k^\dagger]=\delta_{jk}.
\end{equation}
For each mode we introduce the quadrature operators
\begin{equation}
\hat x_j=\hat a_j+\hat a_j^\dagger,
\qquad
\hat p_j=\frac{\hat a_j-\hat a_j^\dagger}{i},
\end{equation}
which play the role of generalized position and momentum. Linear combinations of these operators are routinely measured in quantum optics via homodyne detection~\cite{Leonhardt1995}. Collecting all quadratures into a single vector,
\begin{equation}
\hat{\bm{\xi}}=(\hat x_1,\hat p_1,\hat x_2,\hat p_2,\ldots)^\top,
\end{equation}
the corresponding covariance matrix is defined as
\begin{equation}
\sigma_{mn}
=
\frac{1}{2}
\langle
\hat \xi_m \hat \xi_n+\hat \xi_n \hat \xi_m
\rangle
-
\langle \hat \xi_m\rangle\langle \hat \xi_n\rangle .
\end{equation}

In this work, the elements of the covariance matrix constitute the relevant observables. This choice is natural experimentally, since they can be accessed from homodyne statistics, and it is central theoretically, because Gaussian transformations act linearly on first and second moments. As a result, covariance-based Gaussian reservoirs provide a concrete realization of the linear-reservoir setting introduced in Sec.~II.

\subsection{Gaussian states and dynamics}
\begin{figure*}[t]
\centering
\includegraphics[trim=0cm 8.0cm 0cm 0cm,clip=true,width=0.95\textwidth]{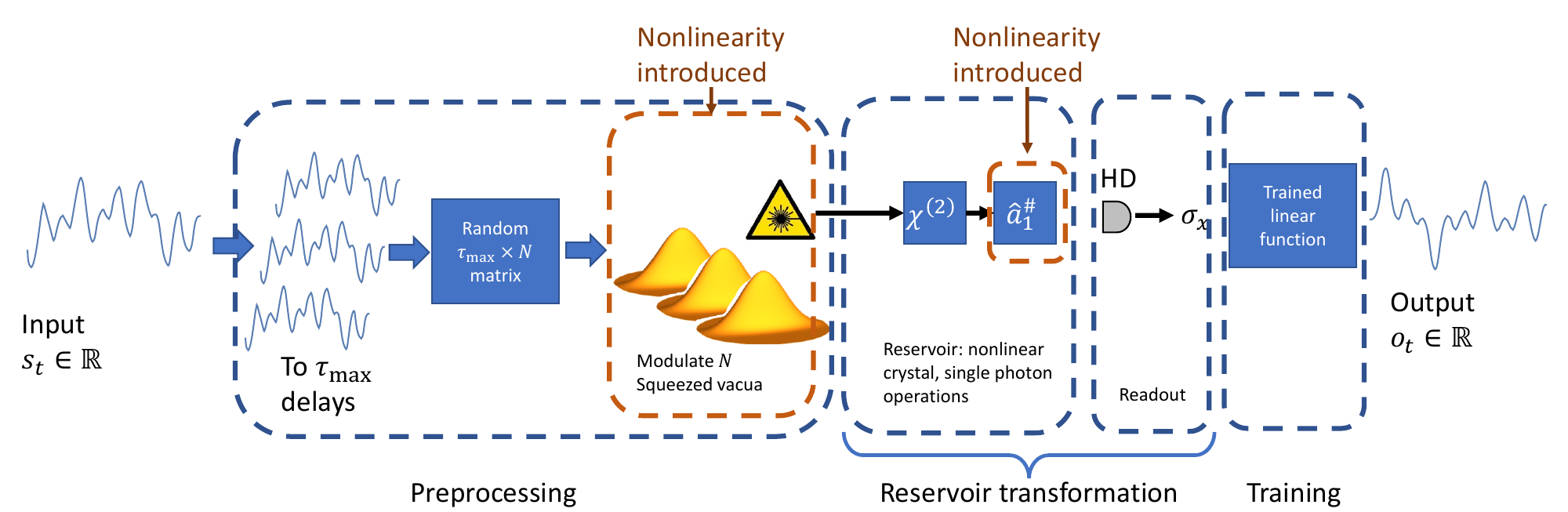}
\caption{\label{fig:scheme} \textbf{The considered scheme.} Memory is classical, which keeps the scheme within reach of experiments. The output is formed from the elements of the covariance matrix of $x$-quadratures $\sigma_x$. Nonlinearity in input is introduced by the input dependent modulation of squeezed states and again in the de-Gaussification of the reservoir state via photon addition or subtraction on mode 1 (assumed to be deterministic, or post-selected).}
\end{figure*}
Gaussian states are states whose Wigner function is Gaussian in phase space. Equivalently, they are completely characterized by the first moments and covariance matrix of the quadratures. Gaussian operations are generated by Hamiltonians at most quadratic in the field operators, and they map Gaussian states to Gaussian states.

At the level of covariance matrices, a Gaussian transformation acts as
\begin{equation}
\sigma \mapsto S \sigma S^\top,
\label{eq:gaussian_cov_update}
\end{equation}
where $S$ is a real symplectic matrix. After vectorization, Eq.~\eqref{eq:gaussian_cov_update} becomes a linear map on the vector of covariance-matrix elements. Consequently, covariance-based Gaussian reservoirs belong to the class of linear reservoirs considered in Theorem~\ref{thm:linear_bottleneck}. The Gaussian bound studied below is therefore a direct specialization of the general linear-reservoir bottleneck.

This observation is important conceptually. The limitation derived in Sec.~II is not specific to one photonic architecture; rather, Gaussian CV reservoirs provide an analytically tractable and experimentally mature setting in which the bottleneck becomes especially transparent.

\subsection{Photonic reservoir schemes}

We consider a photonic reservoir driven by an input time series through squeezed-vacuum encoding. Each input sample is first mapped to parameters of single-mode squeezed states, producing an input-dependent multimode Gaussian state. Correlations between modes are then generated by a fixed Gaussian multimode transformation, which may be implemented, for example, by an optical network \cite{centrone2021cost} or by propagation through a $\chi^{(2)}$ medium. This keeps the architecture close in spirit to existing Gaussian photonic reservoirs, such as the ones in Refs.~\cite{nokkala2021gaussian,nokkala2021high,garcia2023scalable,garcia2023squeezing,nokkala2023online,nokkala2024retrieving,paparelle2025experimental}, while introducing the minimal non-Gaussian ingredient needed to surpass the Gaussian bound; related CV photonic variational architectures have also combined Gaussian processing with photon-subtraction-based nonlinearities and homodyne readout outside the reservoir-computing setting \cite{stornati2024variational,krasimirov2025hardware,polo2025non}.

If the reservoir contains only the above ingredients, then the full transformation from the encoded input covariance matrix to the measured covariance matrix remains Gaussian and therefore falls within the scope of Theorem~\ref{thm:linear_bottleneck}. In particular, no new delay-resolved expressive power can be generated beyond what is already present in the encoded input.

To break this bottleneck, we introduce a minimal non-Gaussian extension based on conditional single-photon operations. Concretely, after the Gaussian multimode transformation, we apply mode-selective photon subtraction or addition,
\begin{equation}
\hat a_j
\qquad \text{or} \qquad
\hat a_j^\dagger,
\end{equation}
conditioned on successful heralding events. These operations are intrinsically non-Gaussian and have been demonstrated experimentally in multimode photonic settings. Their role in the present context is to generate observables that are no longer constrained by the Gaussian linear-reservoir bound.

The setup is presented in Fig.~\ref{fig:scheme}. Memory of recent input history is achieved with classical means; here an approach based on preprocessing is depicted. We return to the details in the next Section.

This setup is attractive for two reasons. First, it remains close to experimentally established Gaussian photonic reservoir schemes: Gaussian state preparation, Gaussian mode mixing, and homodyne detection are all standard ingredients. Second, it introduces the non-Gaussian resource in a minimal and operationally meaningful way, allowing us to directly test whether single-photon operations can produce an excess capacity above the Gaussian bound. 

\subsection{Readout and scope of the analysis}

Throughout this work, the readout is constructed from homodyne data by estimating the covariance matrix of selected quadratures and training a linear combination of its elements. This keeps the readout procedure close to existing continuous-variable reservoir-computing experiments and ensures that any excess delay-resolved capacity can be attributed to the non-Gaussian processing itself rather than to nonlinear post-processing at the output stage.

The remainder of the paper specializes the general theorem of Sec.~II to this CV setting. We first quantify the Gaussian bound for covariance-based reservoirs and then study how conditional single-photon operations break it in practice.

\section{The Gaussian bound and how to break it}

We now turn to the central numerical question of this work: can experimentally accessible non-Gaussian operations generate an information-processing capacity that is forbidden to Gaussian reservoirs? Theorem~\ref{thm:linear_bottleneck} implies that, for covariance-based Gaussian continuous-variable reservoirs, the capacity available at any fixed delay is bounded by the capacity already present in the encoded input. In the memoryless case this yields a Gaussian limit that scales only linearly with reservoir size, whereas the absolute capacity allowed by the number of trained observables scales quadratically. This is deduced counting the maximum number of independent elements of the covariance matrix. Any observed excess above the Gaussian limit therefore signals a genuine breakdown of the Gaussian bottleneck.

The following is a direct consequence of Theorem~\ref{thm:linear_bottleneck}:

\begin{corollary}[]
\label{cor:witness}
Under the assumptions of Theorem~\ref{thm:linear_bottleneck} and its Gaussian specialization, any observed delay-resolved capacity exceeding the corresponding Gaussian bound rules out a Gaussian explanation of the effective reservoir transformation.
\end{corollary}

This witnessing strategy is operational, in the sense that it is formulated directly in terms of input-output processing data and the measured observables used for the readout. It does not require full state tomography or access to the internal reservoir state. At the same time, its scope should be interpreted carefully. What is ruled out is a \emph{Gaussian explanation of the effective reservoir transformation} under the stated assumptions. In particular, the witness does not by itself distinguish whether non-Gaussianity originates from the reservoir dynamics, the measurement stage, or another part of the effective processing map not included in preprocessing. It is therefore best understood as a black-box witness of non-Gaussian processing at the level relevant for information processing capacity.

Throughout the rest of this work, the input state is a product state of single mode squeezed vacua. Each mode has a fixed squeezing magnitude $r=0.75$ whereas the phase of squeezing $\varphi$ is modulated by the input as shown in the preprocessing box of Fig.~\ref{fig:scheme}. Concretely, $\varphi_t=\sum_{\tau=0}^{\tau_\mathrm{max}} c_\tau s_{t-\tau}$, where each $c_\tau\in[0.1/\tau_\mathrm{max},2\pi/\tau_\mathrm{max}]$ is distributed uniformly, and chosen independently for each delay and for each mode. As noted before, the input $s_t\in[-1,1]$ is uniformly distributed.

\subsection{Quantum extreme learning machine}

We begin with the memoryless case by setting $\tau_\mathrm{max}=0$, corresponding to a quantum extreme learning machine. The input is divided into training and testing phases of $5000$ and $1000$ timesteps, respectively. The observables from the training phase are used only to optimize the weights of the linear readout function. Then, data is collected during the testing phase to calculate capacity.

Figure~\ref{fig:qelm_capacities} summarizes the behavior in this regime. The upper panel shows the total capacity as a function of reservoir size $N$, together with the Gaussian limit and the absolute limit set by the number of available observables. The purely Gaussian case remains below the predicted Gaussian bound, as expected. By contrast, the introduction of conditional single-photon operations drives the observed capacity into the region forbidden to Gaussian reservoirs. Moreover, the violation becomes stronger as additional operations are introduced, establishing them as a genuine computational resource.

The lower panel of Fig.~\ref{fig:qelm_capacities} clarifies why saturating the absolute limit rapidly becomes difficult even after the Gaussian bound is broken. It shows how the total capacity is distributed across polynomial degrees in the memoryless case. Since the total capacity is the area under this curve, its growth requires support on progressively higher-degree basis functions at the same delay. In other words, once memory is absent, all nonlinear expressive power must be concentrated into a single delay sector. This provides a natural explanation for the observed slow growth: non-Gaussian operations open access to functions that are inaccessible to Gaussian reservoirs, but concentrating all expressive power into one delay becomes increasingly inefficient as the reservoir grows.

By Corollary~\ref{cor:witness}, the excess above the Gaussian bound already constitutes an operational witness of non-Gaussian processing in the effective reservoir transformation.

\begin{figure}[h]
\centering
\includegraphics[width=\linewidth]{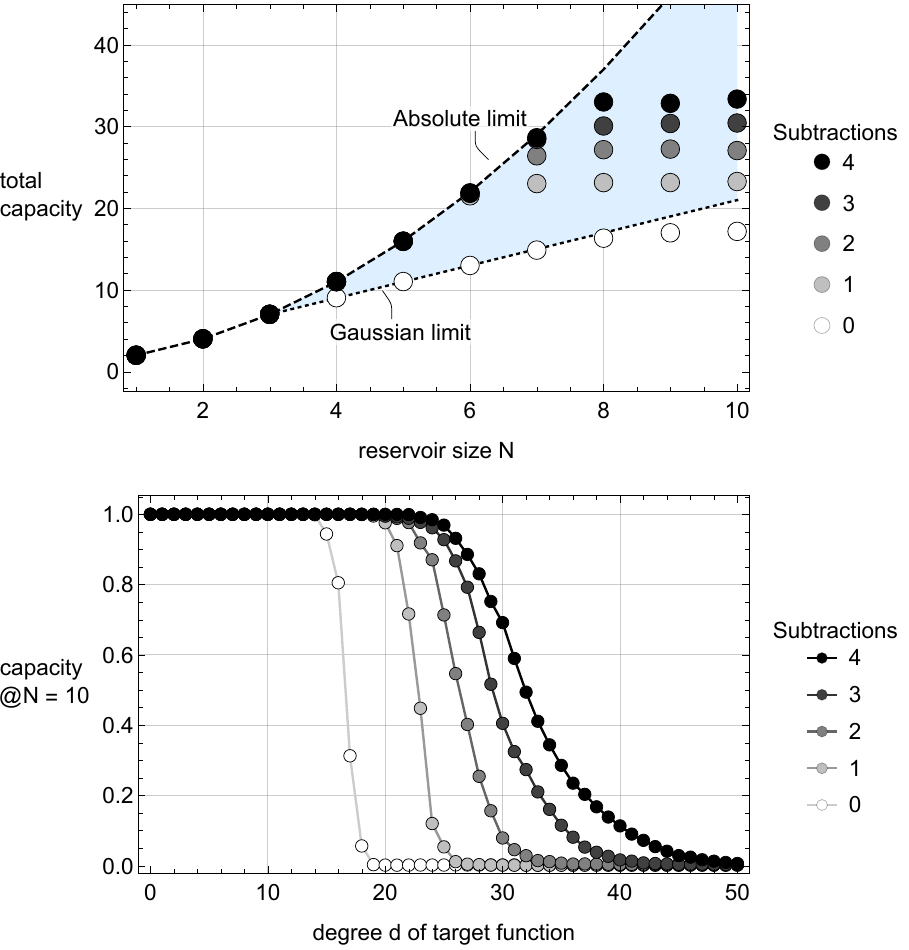}
\caption{\label{fig:qelm_capacities}
\textbf{Breaking the Gaussian limit in the case of a memoryless reservoir and memoryless preprocessing.} Top panel: total capacity as a function of the number of optical modes $N$ for a purely Gaussian reservoir and for reservoirs augmented by conditional single-photon operations. The dotted curve indicates the Gaussian limit, while the dashed curve indicates the absolute limit set by the number of trained weights. The shaded region is therefore inaccessible to Gaussian reservoirs but becomes reachable once non-Gaussian operations are introduced. Bottom panel: breakdown of the capacity among polynomial degrees. Since the total capacity is the area under this curve, approaching the absolute limit requires support on increasingly high-degree basis functions, making saturation progressively harder already at modest reservoir sizes. All results are averaged over $80$ realizations.}
\end{figure}

\subsection{Quantum reservoir computer\label{sec:toQRC}}

The previous result suggests that the slow growth observed in the memoryless case is not due only to weak non-Gaussianity, but also to the structure of the IPC basis itself: all nonlinear expressive power must be concentrated at one delay. Introducing memory changes this dramatically. When multiple delays are available, the reservoir can realize products of lower-degree functions at different delays instead of relying exclusively on very high-degree functions at a single delay.

With memory present, the input $s$ should be divided into three phases. The first is the preparation or washout phase which establishes the proper dependency on recent input history. The other two are the normal training and testing phases. We have used $100$, $5000$ and $1000$ timesteps for these phases respectively.

\begin{figure}[h]
\centering
\includegraphics[width=\linewidth]{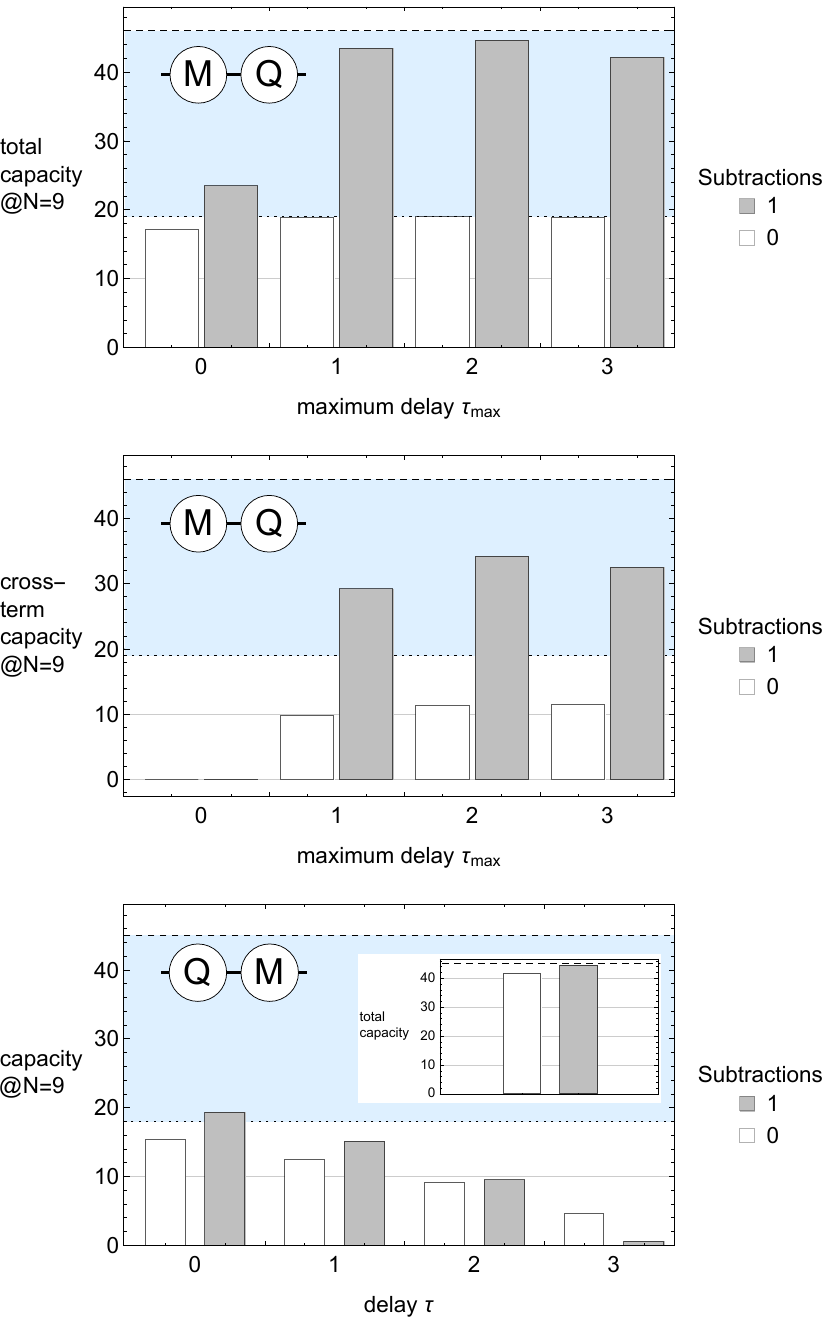}
\caption{\label{fig:qrc_capacities}
\textbf{Breaking the Gaussian limit when there is memory.} Top panel: total capacity as a function of memory length $\tau_\mathrm{max}$ when memory is provided by the preprocessing. Relative to the memoryless case ($\tau_\mathrm{max}=0$), allowing multiple delays strongly alleviates the high-degree bottleneck and enables the non-Gaussian reservoir to approach the absolute limit much more efficiently. Middle panel: capacity of cross-terms only. When preprocessing provides memory, even the Gaussian case is able to have capacity from nonlinear functions featuring multiple delays. This leaves the bound unchanged. Bottom panel: capacity as a function of delay when memory is provided after the quantum reservoir. Here this is accomplished by a simple leaky neuron. The Gaussian bound restricts delay-resolved capacity, but not the total capacity. Neither Gaussian nor non-Gaussian cases have cross-term capacity. All results are averaged over $80$ realizations.}
\end{figure}

Figure~\ref{fig:qrc_capacities} shows the corresponding behavior when memory is introduced through classical preprocessing by setting $\tau_\mathrm{max}>0$. Focusing on the top panel, the contrast with the memoryless case is striking. Even a small number of nontrivial delays substantially enhances the accessible capacity, and the non-Gaussian reservoir approaches the absolute limit much more efficiently than in the memoryless setting. This supports the interpretation suggested by Fig.~\ref{fig:qelm_capacities}: memory does not merely increase the total capacity quantitatively, but changes how that capacity can be assembled from the IPC basis. Once multiple delays are available, the reservoir no longer needs to rely so heavily on extremely high-degree functions at a single delay.

Another point that deserves special analysis is the impact of cross terms. In this context, a cross-term is a nonlinear function featuring multiple delays; for example, $P_{d_1}^{\tau_1}(s_t)P_{d_2}^{\tau_2}(s_t)$ is a cross-term when $d_1,d_2>0$ and $\tau_1\neq\tau_2$. Their significance comes from the fact that normally, a Gaussian reservoir with coherent memory is known to not have them when using covariances \cite{nokkala2021gaussian,nokkala2021high}. Indeed, Gaussian reservoirs should have capacities of a (linear) reservoir with memory, as in Fig.~\ref{fig:limit}. 

Unlike one might expect based on earlier literature~\cite{nokkala2021gaussian,nokkala2021high}, the presence of cross-terms in the memory does not rule out Gaussianity of the reservoir, nor does it mean the Gaussian bound does not apply. Introducing the memory through preprocessing, as done here, means that case A of Theorem~\ref{thm:linear_bottleneck} applies, constraining the total capacity, allowing in particular the presence of cross-terms. The middle panel of Fig.~\ref{fig:qrc_capacities} shows the contributions from cross-terms only. These are included in the total capacity of the top panel along with the diagonal terms of the form $P_{d}^{\tau}(s_t)$.

The bottom panel of Fig.~\ref{fig:qrc_capacities} shows the case where memory is placed after the quantum reservoir. In fact, from the standpoint of Theorem~\ref{thm:linear_bottleneck} this situation falls into the category where a reservoir has memory, leading to bounded delay-resolved capacities. In the case at hand, the memory is implemented with a simple leaky neuron. Concretely, given measured reservoir observables collected in state vector $x_t$ with $\tau_\mathrm{max}=0$, we form the final output vector $o_t$ by setting $o_t=\rho v^\top o_{t-1}+(1-\rho)x_t$. Here, $v$ is a random vector with elements uniformly and independently distributed in $[0,1]$. We have set $\rho=0.001$, which concentrates the capacity to small delays and leads to the breaking of the bound.

\subsection{Finite resource visibility of non-Gaussianity}

The previous results were obtained in the ideal limit in which expectation values are evaluated exactly. In practice, however, the covariance matrix is estimated from a finite ensemble size $M$ of homodyne outcomes. It is therefore essential to check whether the non-Gaussian advantage, and with it the witness based on excess capacity, remains visible under realistic finite-resource conditions.

The details of the statistical noise model can be found in Appendix~\ref{app:finite_resources}. Essentially, we assume that the covariance matrix elements are affected by Gaussian noise of fixed intensity. More precisely, the reconstructed covariance matrices follow a Wishart distribution~\cite{wishart1928generalised}.  Despite its simplicity, this model is expected to be accurate in the large sample limit, and is exact for Gaussian reservoirs.

\begin{figure}[t]
\centering
\includegraphics[width=\linewidth]{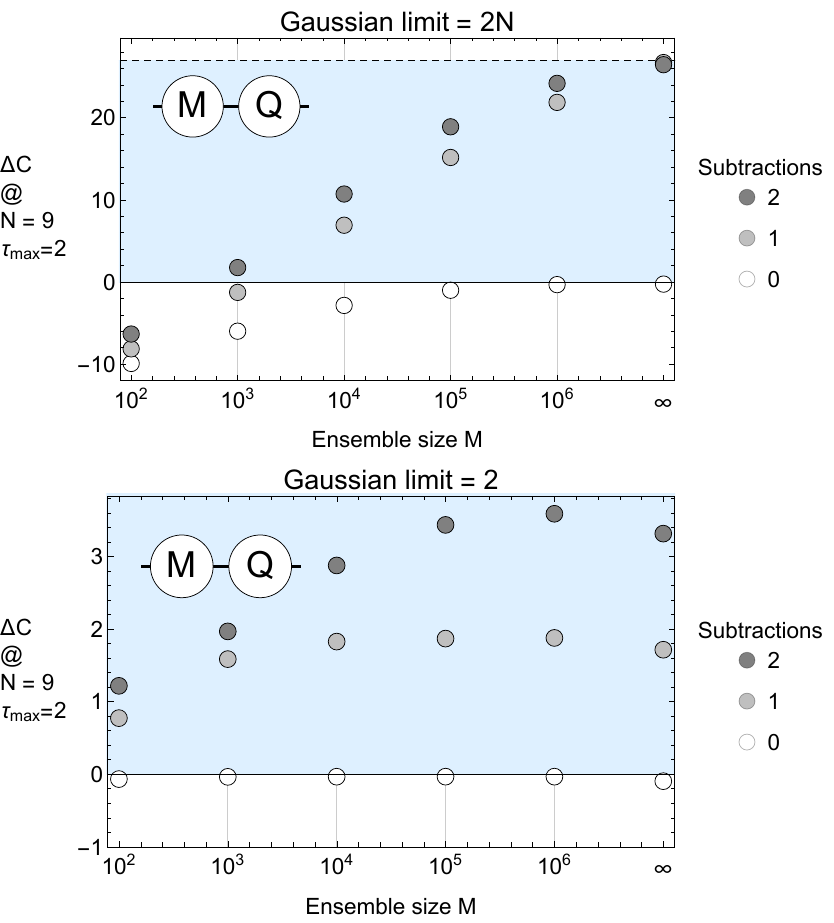}
\caption{\label{fig:finite_resources}
Witnessing the non-Gaussianity of the reservoir with finite resources. Excess capacity $\Delta C$ is the difference between observed capacity and the Gaussian limit. Top panel: $\Delta C$ as a function of ensemble size $M$. Excess capacity degrades gracefully as statistical noise in the observables increases, but eventually the noise washes it out. Bottom panel: excess capacity when the Gaussian limit is artificially restricted improves finite-resource visibility. Unlike elsewhere, here all input modes have identical encoding, turning the Gaussian limit constant. This decreases capacity but allows a finite $\Delta C$ at lower ensemble sizes. All results are averaged over 80 realizations.
}
\end{figure}

Unlike elsewhere in this work, where it would not have observable consequences, here we use thresholding, which means we set small capacities to zero, using the approach of Ref.~\cite{dambre2012information}. This is to avoid overestimating the actual capacity and in this way avoid false positives. The probability to interpret a vanishing capacity as finite, which controls the value of the threshold, has been set to $p=10^{-10}$. We also set the preparation and training phase to $100$ and $10000$ timesteps, respectively. 

Figure~\ref{fig:finite_resources} compares the ideal and finite-ensemble cases The top panel shows the behavior of the excess capacity $\Delta C$ as a function of ensemble size $M$. As expected, finite sampling reduces the measured capacity in both Gaussian and non-Gaussian cases. Nevertheless, the qualitative separation between the Gaussian and non-Gaussian reservoirs persists: non-Gaussian operations still shift capacity toward higher degrees and maintain a larger total capacity over a broad range of reservoir sizes. This shows that the breaking of the Gaussian bound is not merely an idealized infinite-sampling effect, but can remain operationally visible in realistic regimes.

The bottom panel shows the case where the Gaussian limit has been artificially restricted by using identical phase modulation for all modes. In this case the Gaussian limit is just two, provided by one sine and one cosine whose arguments are directly proportional to the input. These functions come directly from the covariance matrix of a single squeezed mode. We note that $\Delta C$ appears to even increase when going to finite resources case. This may be caused by the sampling noise playing the role of regularization \cite{nokkala2021high}.

Interestingly, while this decreases capacity in all cases it also substantially improves the visibility of a finite $\Delta C$. This could be particularly relevant when the goal is to witness the non-Gaussianity of the reservoir, rather than use it for QRC.

\section{Discussion and outlook}

The main result of this work is the identification of a hidden bottleneck in information processing capacity when the relevant observables of a reservoir admit a finite-dimensional linear update. Under these conditions, the reservoir cannot create new delay-resolved expressive power on its own: it can only reshape what preprocessing has already made available. In continuous-variable quantum reservoir computing based on covariance-matrix observables, this yields a concrete Gaussian bound. The importance of this result is not only that it formalizes an intuitive limitation of Gaussian reservoirs, but also that it reveals where the limitation actually lies. The bottleneck is delay resolved and may therefore remain invisible in aggregate performance measures or in the total information processing capacity.

Our numerical results show that conditional single-photon operations provide a direct route beyond this limitation. In the memoryless setting they generate capacities that enter the region forbidden to Gaussian reservoirs, thereby establishing them as a genuine computational resource for continuous-variable quantum reservoir computing. At the same time, the memoryless case also reveals a second feature: breaking the Gaussian bound is not by itself enough to saturate the absolute capacity limit. The observed slow growth suggests that when only one delay is available, increasing total capacity requires access to progressively higher-degree functions at that same delay, making the memoryless setting intrinsically difficult to scale.

Introducing memory changes this picture qualitatively. Once multiple delays are available, expressive power can be assembled from products of lower-degree functions supported on different delays instead of being forced into very high degrees at a single delay. Our numerical evidence suggests that this strongly alleviates the scaling bottleneck and allows non-Gaussian resources to be used much more efficiently. In this sense, the paper points to a practical design principle: when nonlinear processing is costly, distributing it across multiple delays can be far more effective than concentrating it into a memoryless architecture.

The witness interpretation of excess capacity should be understood in this operational sense. Under the assumptions of Theorem~\ref{thm:linear_bottleneck}, any observed excess above the corresponding Gaussian bound rules out a Gaussian explanation of the effective reservoir transformation. The witness is therefore formulated directly in terms of input-output processing data and the measured observables used for the readout, without requiring full state tomography or access to the internal reservoir state. At the same time, its scope is limited: it is a sufficient but not necessary witness, and it certifies non-Gaussian \emph{processing} under the stated assumptions rather than providing a universal characterization of all possible sources of non-Gaussianity in the experiment.

These results also clarify the role of Gaussian reservoirs. Being subject to the bound does not make a reservoir useless. For many tasks, the available delay-resolved expressive power may still be sufficient, and Gaussian reservoirs remain attractive because of their simplicity, scalability, and robustness. Moreover, Gaussian mode mixing can still help redistribute information among accessible observables even when it does not break the bottleneck. The present bound should therefore not be read as a no-go theorem against Gaussian quantum reservoir computing, but rather as a guide to when Gaussian resources are enough and when genuinely non-Gaussian ingredients become necessary.

Several directions remain open. On the theoretical side, it will be important to better understand the scaling bottleneck observed in the memoryless case and to quantify more precisely how memory alleviates it. On the numerical side, the study of larger numbers of conditional single-photon operations may clarify whether the gains eventually saturate or continue to accumulate in experimentally relevant regimes. On the practical side, finite-resource effects deserve further attention, especially in connection with measurement overhead and the tradeoff between direct non-Gaussian processing and Gaussian strategies supplemented by more elaborate post-processing. While continuous-variable Gaussian reservoirs provide an especially clean and experimentally relevant realization of the linear-reservoir bottleneck, the broader theorem suggests that similar architectural limitations may arise in other reservoir platforms whenever the chosen observables admit a finite-dimensional linear update. Finally, Theorem 1 does not link the Gaussianity of the reservoir to bounds on IPC for an arbitrary choice of observables. One may consider, e.g., using higher-order moments \cite{garcia2023squeezing}, or less orthodox choices such as the cumulative distribution function \cite{hahto2025smarter} to extract even substantially more power. Nevertheless, our results have provided a fertile ground for further investigations and may eventually contribute to a formal resource theory.

In summary, our results identify a structural limitation of Gaussian and, more generally, linear reservoir transformations, show how experimentally accessible non-Gaussian operations can overcome it, and provide a concrete framework for using excess capacity as an operational signature of non-Gaussian processing. Together, these results help clarify which resources are genuinely responsible for enhanced expressive power in continuous-variable quantum reservoir computing and suggest concrete directions for the design of more powerful near-term reservoir architectures.

\begin{acknowledgments}
We thank Roberta Zambrini and Valentina Parigi for fruitful discussions. JN gratefully acknowledges financial support from the Research Council of Finland under Project No. 348854. JN  acknowledges The Finnish Quantum Institute (InstituteQ). F.C. acknowledges funding by the European Union (EQC, 101149233, PASQuanS2.1, 101113690), the Government of Spain (Severo Ochoa CEX2019-000910-S, FUNQIP and NextGenerationEU PRTR-C17.I1).
\end{acknowledgments}

\bibliography{references}

\clearpage
\onecolumngrid

\appendix

\section{Reservoir computing\label{app:training}}

Reservoir computing is concerned with transformations between two time series. Let $\mathbf{s}=\{\ldots,s_{t-2},s_{t-1},s_t\}$ be the input time series up to some fixed timestep $t$ and $\mathbf{o}=\{\ldots,o_{t-2},o_{t-1},o_t\}$ be the corresponding output time series. The two are related by some fixed but unknown functional $F$ such that $s_t=F[\mathbf{o}]$. Tasks can then be solved by some well-behaved functional $F$. In this context, well-behaved means that $F$ can be well approximated by some continuous function $f$ of only a finite number of past inputs. With a slight abuse of terminology, such $F$ are called fading memory functions. They can be shown to form a real Hilbert space when the inner product is defined via integration of a product of two functions \cite{dambre2012information}; this leads to a notion of orthogonality and orthonormal bases. This space of fading memory functions is infinite-dimensional.

In QRC, such approximations are achieved by driving a quantum system with $\mathbf{s}$ and post-processing its response to approximate the desired output. A generic QRC scheme consists of an input mechanism, the reservoir, a readout mechanism, and training. The input mechanism is responsible for allowing one to drive the reservoir dynamics based on a classical time series, often by preparing an input state conditioned on it. Logically, the input mechanism plays the role of preprocessing (see Fig.~\ref{fig:limit}). The reservoir realizes an input independent transformation of the input state. It often equips the resulting state with a richer dependency on the input through, e.g., creating correlations. The readout mechanism is responsible for extracting the chosen reservoir observables via suitable (ensemble) measurements. Together, the reservoir and the readout mechanism play the role of the reservoir transformation. We mention in passing that keeping the reservoir fixed but changing the readout mechanism may lead to an even substantially different reservoir transformation \cite{hahto2025smarter}.

Time series processing requires memory of recent input history, which is conventionally provided by the state of the reservoir. This makes QRC somewhat similar to recurrent quantum neural networks. In the limiting case where there is memory only of the most recent input $s_t$, the proper term is quantum extreme learning machine, which is akin to feedforward quantum neural networks. Systems suitable for reservoir computing have the fading memory property, which means that under drive their observables become fading memory functions which then span a subspace in the full function space.  When these functions are (non)linear in input we speak of (non)linear memory. 

Training is achieved via classical post-processing of recorded values of computational nodes. Since this kind of training  involves only classical data, it provides QRC its characteristic engineering freedom and rapid training. Typically one forms the optimal linear combination of the nodes, which amounts to an orthogonal projection of $F$ into the subspace, minimizing the mean squared error
\begin{equation}
\mathrm{NMSE(\bar{\mathbf{o}},\mathbf{o})=\sum_k\frac{(\bar{\mathbf{o}}-\mathbf{o})^2}{\bar{\mathbf{o}}^2}}
\end{equation}
between the target output $\bar{\mathbf{o}}$ and actual output $\mathbf{o}$. Linear combinations of observables are computationally very cheap and consequently offload the bulk of the processing to the reservoir, harnessing its dynamics for nontrivial transformations.

In practice, for a given input and target one divides $\mathbf{s}$ into preparation, training and test phases. The first phase will erase information about the initial state of the reservoir. In the training phase, the observables at each time step are used as rows of a matrix $\mathbf{X}$ and a column vector of all ones is typically added to account for a constant bias term. Weights of some linear combination of reservoir observables can now be given as a column vector $\mathbf{W}$ such that $\mathbf{X}\mathbf{W}=\mathbf{o}$ where $\mathbf{o}$ is the actual output that should approximate a target output $\bar{\mathbf{o}}$. The weights that minimize NMSE can be found by setting \begin{equation}    \mathbf{W}=\mathbf{X}^+\bar{\mathbf{o}}^\top
\end{equation}
where $\mathbf{X}^+$ is the Moore-Penrose pseudoinverse \cite{jaeger2002tutorial,montavon2012neural}. These weights are then  used to assess the NMSE when the reservoir is exposed to previously unseen cases in the testing phase. It should be noted that for any finite length of training phase this only approximates an orthogonal projection due to finite data effects. Regularization could be considered to prevent overfitting, however in QRC statistical noise from a finite ensemble size leads to additive noise in reservoir observables, which in turn can play the role of regularization \cite{nokkala2021high}.

\section{Proof of Theorem~\ref{thm:linear_bottleneck}}
\label{app:proof_linear}

\subsection{Affine reservoir updates}
\label{app:affine}

In the most general case, the linear reservoir transformation may be affine rather than strictly linear,
\begin{equation}
    x_t=\mathbf{A}x_{t-1}+\mathbf{B}g_t+\mathbf{c},
    \label{eq:affine_update}
\end{equation}
where $\mathbf{c}$ is a constant vector. This causes no essential change because the output already includes a freely tunable bias term.

Indeed, define the augmented state and input vectors
\begin{equation}
    \tilde{x}_t=
    \begin{pmatrix}
        x_t\\
        1
    \end{pmatrix},
    \qquad
    \tilde{g}_t=
    \begin{pmatrix}
        g_t\\
        1
    \end{pmatrix},
\end{equation}
and the augmented matrices
\begin{equation}
    \tilde{\mathbf{A}}=
    \begin{pmatrix}
        \mathbf{A} & \mathbf{c}\\
        \mathbf{0}^\top & 1
    \end{pmatrix},
    \qquad
    \tilde{\mathbf{B}}=
    \begin{pmatrix}
        \mathbf{B} & \mathbf{0}\\
        \mathbf{0}^\top & 0
    \end{pmatrix}.
\end{equation}
Then Eq.~\eqref{eq:affine_update} is equivalent to the purely linear update
\begin{equation}
    \tilde{x}_t=\tilde{\mathbf{A}}\tilde{x}_{t-1}+\tilde{\mathbf{B}}\tilde{g}_t.
\end{equation}
Similarly, the readout
\begin{equation}
    o_t=\mathbf{w}^\top x_t+w_0
\end{equation}
can be written as
\begin{equation}
    o_t=\tilde{\mathbf{w}}^\top \tilde{x}_t,
    \qquad
    \tilde{\mathbf{w}}=
    \begin{pmatrix}
        \mathbf{w}\\
        w_0
    \end{pmatrix}.
\end{equation}
Therefore it is sufficient to prove Theorem~\ref{thm:linear_bottleneck} for strictly linear updates. In the remainder of the appendices we omit the tildes for simplicity.

\subsection{Case A: preprocessing with memory, memoryless reservoir}

In this case $\mathbf{A}=\mathbf{0}$ and $g_t$ is allowed to depend on input history $\{s_{t-\tau}\}_{\tau\geq0}$, rather than only on the most recent input $s_t$. The only requirement is fading memory.

The information-processing capacity (IPC) is upper bounded by the number of available linearly independent functions of the input, and in the asymptotic setting of Ref.~\cite{dambre2012information} this bound is saturated for systems with fading memory. The number of linearly independent functions coincides with the rank of the data matrix. Let $\mathbf{X}^{\mathrm{in}}$ be the matrix whose $k$th row is $g_k^\top$, and let $\mathbf{X}$ denote the corresponding matrix after the reservoir transformation. We omit the constant bias column, since it contributes only to the degree-zero direction.

Using Eq.~\eqref{eq:linear_reservoir_update} with $\mathbf{A}=\mathbf{0}$, we obtain
\begin{equation}
    \mathbf{X}=\mathbf{X}^{\mathrm{in}}\mathbf{B}^\top.
\end{equation}
Therefore,
\begin{equation}
    \mathrm{IPC}(\mathbf{X})
    \leq \mathrm{rank}(\mathbf{X})
    =\mathrm{rank}(\mathbf{X}^{\mathrm{in}}\mathbf{B}^\top)
    \leq \min\!\left(\mathrm{rank}(\mathbf{X}^{\mathrm{in}}),\mathrm{rank}(\mathbf{B}^\top)\right)
    \leq \mathrm{rank}(\mathbf{X}^{\mathrm{in}}).
\end{equation}
Hence the total IPC after the reservoir is upper bounded by the IPC already present in preprocessing alone. This proves part (A) of Theorem~\ref{thm:linear_bottleneck}.

\subsection{Auxiliary lemma on fixed delays}

We next formalize the fact that functions supported on different delays do not contribute to the reconstruction of a target at a fixed delay.

\begin{lemma}
\label{lem:fixed_delay}
Let $\tau'$ be fixed and let $z_t=P_d(s_{t-\tau'})$ with $d>0$. Let $\mathcal{U}_{\tau'}$ be the span of the constant function together with all functions of the single variable $s_{t-\tau'}$, and let $\mathcal{V}_{\tau'}$ be the span of all nonconstant functions depending only on delays different from $\tau'$. Then
\begin{equation}
    \mathcal{U}_{\tau'} \perp \mathcal{V}_{\tau'}
\end{equation}
in the IPC inner-product space. In particular, the capacity for reconstructing $z_t$ from any feature space contained in $\mathcal{U}_{\tau'}\oplus\mathcal{V}_{\tau'}$ is equal to the capacity obtained from its projection onto $\mathcal{U}_{\tau'}$ alone.
\end{lemma}

\begin{proof}
The IPC inner product is the expectation value of the product of two fading-memory functions. Since the inputs are independently distributed across delays, the expectation of a product factorizes over distinct delays. For any nonconstant function $v_t\in\mathcal{V}_{\tau'}$ one has
\begin{equation}
    \langle z_t,v_t\rangle
    =\mathbb{E}[z_t v_t]
    =\mathbb{E}[z_t]\mathbb{E}[v_t]
    =0,
\end{equation}
because $z_t=P_d(s_{t-\tau'})$ with $d>0$ has zero mean. The constant function is included in $\mathcal{U}_{\tau'}$, so the only possible non-orthogonal constant direction is already absorbed there. Therefore $\mathcal{U}_{\tau'}$ and $\mathcal{V}_{\tau'}$ are orthogonal. Since $z_t\in\mathcal{U}_{\tau'}$, the orthogonal projection of $z_t$ onto $\mathcal{U}_{\tau'}\oplus\mathcal{V}_{\tau'}$ coincides with its projection onto $\mathcal{U}_{\tau'}$ alone, which proves the claim.
\end{proof}

\subsection{Case B: memoryless preprocessing, reservoir with memory}

In this case $g_t=g(s_t)$ depends only on the most recent input, whereas the reservoir recurrence carries memory. The solution to Eq.~\eqref{eq:linear_reservoir_update} is
\begin{equation}
    x_t=\mathbf{A}^t x_0+\sum_{\tau=0}^{t-1}\mathbf{A}^\tau\mathbf{B}g(s_{t-\tau}),
    \label{eq:unrolled_linear}
\end{equation}
where $x_0$ is an arbitrary initial state. By fading memory,
\begin{equation}
    \lim_{t\rightarrow\infty}\mathbf{A}^t x_0=0.
\end{equation}
Hence in the asymptotic regime the observables are sums of terms depending on one delay at a time.

Fix a delay $\tau'$ and consider the task of reconstructing
\begin{equation}
    z_t=P_d(s_{t-\tau'}).
\end{equation}
By Eq.~\eqref{eq:unrolled_linear}, the feature space generated by $x_t$ is contained in the sum of the space generated by the term
\begin{equation}
    \mathbf{A}^{\tau'}\mathbf{B}g(s_{t-\tau'})
\end{equation}
and a second space containing all contributions from delays different from $\tau'$. By Lemma~\ref{lem:fixed_delay}, only the matching-delay term can contribute to the reconstruction of $z_t$. Therefore
\begin{equation}
   C\!\left(x_t,P_d(s_{t-\tau'})\right)
   \leq
   C\!\left(\mathbf{A}^{\tau'}\mathbf{B}g(s_{t-\tau'}),P_d(s_{t-\tau'})\right).
   \label{eq:delay_bound}
\end{equation}

Since this holds for every degree $d$, the delay-resolved IPC at fixed delay $\tau'$ is upper bounded by the IPC obtained from the correct term alone. Finally, by part (A) already proved above, the linear map $\mathbf{A}^{\tau'}\mathbf{B}$ cannot increase the available IPC, so
\begin{equation}
    \mathrm{IPC}^{(\tau')}(x_t)\leq \mathrm{IPC}^{(\tau')}(g(s_{t-\tau'})).
\end{equation}
This proves part (B) of Theorem~\ref{thm:linear_bottleneck}.

\section{Gaussian continuous-variable reservoirs as an application}
\label{app:gaussian_example}

Gaussian QRC often uses first moments or covariance-matrix elements as computational nodes. We focus here on the covariance-matrix formulation; the case of means is analogous.

Let $\sigma_t$ be an input covariance matrix determined by the input at timestep $t$, and let $\sigma_t^{\mathrm{out}}$ denote the output covariance matrix whose elements are ultimately used for the readout. The input state itself need not be Gaussian. The Gaussian reservoir transformation is determined by a fixed symplectic matrix $\mathbf{S}$.

\subsection{Case A}

In the memoryless-reservoir case,
\begin{equation}
    \sigma_t^{\mathrm{out}}=\mathbf{S}\sigma_t\mathbf{S}^\top.
\end{equation}
After vectorization this becomes
\begin{equation}
    \mathrm{vec}(\sigma_t^{\mathrm{out}})
    =
    (\mathbf{S}\otimes\mathbf{S})\,\mathrm{vec}(\sigma_t),
\end{equation}
where $\otimes$ is the Kronecker product. Thus Eq.~\eqref{eq:linear_reservoir_update} applies with
\begin{equation}
    x_t=\mathrm{vec}(\sigma_t^{\mathrm{out}}),\qquad
    g_t=\mathrm{vec}(\sigma_t),\qquad
    \mathbf{A}=\mathbf{0},\qquad
    \mathbf{B}=\mathbf{S}\otimes\mathbf{S}.
\end{equation}

\subsection{Case B}

In the memoryful-reservoir case, write
\begin{equation}
    \sigma_t^{\mathrm{out}}=\mathbf{S}\left(\sigma^R_{t-1}\oplus\sigma_t\right)\mathbf{S}^\top,
\end{equation}
where $\sigma^R_{t-1}$ is the reservoir covariance matrix and $\oplus$ denotes the direct sum. The reservoir modes evolve according to
\begin{equation}
    \sigma_t^R=\mathbf{P}_R \sigma_t^{\mathrm{out}}\mathbf{P}_R^\top,
\end{equation}
where $\mathbf{P}_R$ projects onto the reservoir modes.

To avoid confusion with the abstract matrices $\mathbf{A}$ and $\mathbf{B}$, decompose the symplectic matrix as
\begin{equation}
    \mathbf{S}=
    \begin{pmatrix}
        \mathbf{M} & \mathbf{N}\\
        \mathbf{P} & \mathbf{Q}
    \end{pmatrix},
\end{equation}
where $\mathbf{M}$ has the dimensions of $\sigma^R_{t-1}$ and $\mathbf{Q}$ those of $\sigma_t$. Then the reservoir covariance update is
\begin{equation}
    \sigma_t^R=\mathbf{M}\sigma_{t-1}^R\mathbf{M}^\top+\mathbf{N}\sigma_t\mathbf{N}^\top.
\end{equation}
After vectorization,
\begin{equation}
    \mathrm{vec}(\sigma_t^R)
    =
    (\mathbf{M}\otimes\mathbf{M})\,\mathrm{vec}(\sigma_{t-1}^R)
    +
    (\mathbf{N}\otimes\mathbf{N})\,\mathrm{vec}(\sigma_t).
\end{equation}
Hence Eq.~\eqref{eq:linear_reservoir_update} applies with
\begin{equation}
    x_t=\mathrm{vec}(\sigma_t^R),\qquad
    g_t=\mathrm{vec}(\sigma_t),\qquad
    \mathbf{A}=\mathbf{M}\otimes\mathbf{M},\qquad
    \mathbf{B}=\mathbf{N}\otimes\mathbf{N}.
\end{equation}

If the output is formed not only from reservoir modes but from a larger subset of modes, or even from the entire $\sigma_t^{\mathrm{out}}$, the same bound still applies. The reason is that the full output covariance matrix is obtained from the block-diagonal matrix $\sigma^R_{t-1}\oplus\sigma_t$ by the invertible linear transformation induced by $\mathbf{S}$, and invertible linear transformations preserve the span of functions. Since IPC is determined by that span, the bound is unchanged.

Finally, Gaussian measurements can be reduced to Gaussian unitaries, auxiliary Gaussian modes, and homodyne detection. These additional Gaussian transformations can therefore be absorbed into the same symplectic description without affecting the bound. Dissipation and decoherence can only decrease accessible IPC, and therefore do not invalidate the bound.

\section{Evaluation of non-Gaussian observables}
\label{app:wick}

In this appendix we describe how the observables of the non-Gaussian reservoirs are evaluated exactly, without truncating the Hilbert space. The key point is that the states used in the numerics are obtained by applying a finite number of creation and annihilation operators to an underlying Gaussian state. This covers both conditional single-photon subtraction and addition on equal footing, and more generally provides the elementary building blocks that appear in perturbative expansions of weak non-Gaussian processing.

Indeed, if a weak non-Gaussian potential is generated by a unitary of the form
\begin{equation}
\hat U_{\mathrm{nG}}=\exp[-i\epsilon V(\hat x,\hat p)],
\qquad \epsilon\ll 1,
\end{equation}
where $V$ is a polynomial in the quadratures, then~\cite{polyapprox}
\begin{equation}
\hat U_{\mathrm{nG}} = \mathbb{I} - i\epsilon V(\hat x,\hat p)+O(\epsilon^2).
\end{equation}
Since $\hat x_j$ and $\hat p_j$ are linear combinations of $\hat a_j$ and $\hat a_j^\dagger$, the correction term is a finite sum of monomials in creation and annihilation operators. The formalism developed below therefore applies directly to the experimentally relevant photon-added and photon-subtracted states considered in this work, and more broadly to the basic terms appearing in weak non-Gaussian expansions.

\subsection{Non-Gaussian states built from Gaussian ones}

Let $\hat\rho_G$ be an $N$-mode Gaussian state and let
\begin{equation}
\hat O = \prod_{r=1}^{n} \hat a_{s_r}^{\#_r},
\qquad
\#_r\in\{\cdot,\dagger\},
\end{equation}
be a finite ordered product of annihilation and creation operators acting on selected modes $s_r$. The corresponding non-Gaussian state is
\begin{equation}
\hat\rho_{\mathrm{nG}}
=
\frac{\hat O \hat\rho_G \hat O^\dagger}
{\Tr\!\left[\hat O \hat\rho_G \hat O^\dagger\right]}
=
\frac{\hat O \hat\rho_G \hat O^\dagger}{K},
\label{eq:rho_nG_appendix}
\end{equation}
with normalization
\begin{equation}
K=\Tr\!\left[\hat O \hat\rho_G \hat O^\dagger\right]
=\Tr\!\left[\hat O^\dagger \hat O \hat\rho_G\right].
\end{equation}
The expectation value of an observable $\hat M$ in the de-Gaussified state is therefore
\begin{equation}
\langle \hat M\rangle_{\mathrm{nG}}
=
\frac{\Tr\!\left[\hat O^\dagger \hat M \hat O\, \hat\rho_G\right]}{K}.
\label{eq:nG_expectation_appendix}
\end{equation}

In the simulations, $\hat M$ is always a polynomial of degree at most two in the quadrature operators, since the readout is ultimately based on first and second moments. Equation~\eqref{eq:nG_expectation_appendix} therefore reduces the problem to evaluating Gaussian expectation values of finite products of ladder operators.

The most general Gaussian state may have nonzero first moments. In that case one may write
\begin{equation}
\hat a_j = \Delta \hat a_j + \alpha_j,
\qquad
\alpha_j = \langle \hat a_j\rangle_G,
\end{equation}
and similarly for $\hat a_j^\dagger$, where the centered operators $\Delta\hat a_j$ have vanishing expectation value in $\hat\rho_G$. The evaluation of Eq.~\eqref{eq:nG_expectation_appendix} can then be reduced to the zero-displacement case by expanding all products in centered operators and c-number displacements.

In the present work, the encoded Gaussian states have zero displacement, so the explicit formulas below are given in that setting. This is the only case used in the numerics.

\subsection{Wick expansion}

For a zero-mean Gaussian state, expectation values of odd numbers of centered field operators vanish, while even moments factorize according to Wick's theorem. Let $\hat b_1,\ldots,\hat b_m$ be operators each equal to some $\hat a_j$ or $\hat a_j^\dagger$. Then
\begin{equation}
\Tr[\hat b_1\cdots \hat b_{2q+1}\hat\rho_G]=0,
\end{equation}
and
\begin{equation}
\Tr[\hat b_1\cdots \hat b_{2q}\hat\rho_G]
=
\sum_{\mathcal P}
\prod_{(u,v)\in\mathcal P}
\Tr[\hat b_u \hat b_v \hat\rho_G],
\label{eq:wick_appendix}
\end{equation}
where the sum runs over all perfect matchings $\mathcal P$ of the set $\{1,\ldots,2q\}$. Hence every required expectation value reduces to a sum over products of pair contractions.

The computational cost is determined by the number of such matchings, which grows as $(2q-1)!!$. In the regime considered here this remains tractable because only a small number of non-Gaussian operations is applied.

The only building blocks needed in Eq.~\eqref{eq:wick_appendix} are the pair contractions
\begin{equation}
\langle \hat a_j \hat a_k\rangle_G,\quad
\langle \hat a_j^\dagger \hat a_k^\dagger\rangle_G,\quad
\langle \hat a_j^\dagger \hat a_k\rangle_G,\quad
\langle \hat a_j \hat a_k^\dagger\rangle_G.
\end{equation}
To express them in terms of the Gaussian covariance matrix, it is convenient to define the quadrature blocks
\begin{equation}
\sigma^{xx}_{jk} := \sigma_{2j-1,\,2k-1},\qquad
\sigma^{xp}_{jk} := \sigma_{2j-1,\,2k},
\end{equation}
\begin{equation}
\sigma^{px}_{jk} := \sigma_{2j,\,2k-1},\qquad
\sigma^{pp}_{jk} := \sigma_{2j,\,2k},
\end{equation}
using the quadrature ordering introduced in the main text,
\begin{equation}
\hat{\bm{\xi}}=(\hat x_1,\hat p_1,\hat x_2,\hat p_2,\ldots)^\top.
\end{equation}
With
\begin{equation}
\hat a_j = \frac{\hat x_j+i\hat p_j}{2},
\qquad
\hat a_j^\dagger = \frac{\hat x_j-i\hat p_j}{2},
\end{equation}
and the commutation relation $[\hat x_j,\hat p_k]=2i\delta_{jk}$, one obtains
\begin{equation}
\langle \hat a_j \hat a_k\rangle_G
=
\frac{1}{4}
\left(
\sigma^{xx}_{jk}
-\sigma^{pp}_{jk}
+i\bigl(\sigma^{xp}_{jk}+\sigma^{px}_{jk}\bigr)
\right),
\label{eq:Iaa_appendix}
\end{equation}
\begin{equation}
\langle \hat a_j^\dagger \hat a_k^\dagger\rangle_G
=
\frac{1}{4}
\left(
\sigma^{xx}_{jk}
-\sigma^{pp}_{jk}
-i\bigl(\sigma^{xp}_{jk}+\sigma^{px}_{jk}\bigr)
\right),
\label{eq:Iadagadag_appendix}
\end{equation}
\begin{equation}
\langle \hat a_j^\dagger \hat a_k\rangle_G
=
\frac{1}{4}
\left(
\sigma^{xx}_{jk}
+\sigma^{pp}_{jk}
+i\bigl(\sigma^{xp}_{jk}-\sigma^{px}_{jk}\bigr)
-2\delta_{jk}
\right),
\label{eq:Iadag_a_appendix}
\end{equation}
\begin{equation}
\langle \hat a_j \hat a_k^\dagger\rangle_G
=
\frac{1}{4}
\left(
\sigma^{xx}_{jk}
+\sigma^{pp}_{jk}
-i\bigl(\sigma^{xp}_{jk}-\sigma^{px}_{jk}\bigr)
+2\delta_{jk}
\right).
\label{eq:Ia_adag_appendix}
\end{equation}
Equations~\eqref{eq:Iaa_appendix}--\eqref{eq:Ia_adag_appendix} fully determine all Gaussian moments needed in the Wick expansion.

\subsection{From ladder-operator moments to the covariance matrix}

The observables used in the simulations are entries of the covariance matrix of selected quadratures in the non-Gaussian state. These are reconstructed from Eq.~\eqref{eq:nG_expectation_appendix} by choosing $\hat M$ to be linear or quadratic in the ladder operators and then converting to quadrature moments.

The first moments are
\begin{equation}
\langle \hat x_j\rangle_{\mathrm{nG}}
=
\langle \hat a_j\rangle_{\mathrm{nG}}
+
\langle \hat a_j^\dagger\rangle_{\mathrm{nG}},
\end{equation}
\begin{equation}
\langle \hat p_j\rangle_{\mathrm{nG}}
=
\frac{\langle \hat a_j\rangle_{\mathrm{nG}}
-
\langle \hat a_j^\dagger\rangle_{\mathrm{nG}}}{i}.
\end{equation}
The second moments follow from
\begin{equation}
\langle \hat x_j \hat x_k\rangle_{\mathrm{nG}}
=
\langle \hat a_j \hat a_k\rangle_{\mathrm{nG}}
+
\langle \hat a_j \hat a_k^\dagger\rangle_{\mathrm{nG}}
+
\langle \hat a_j^\dagger \hat a_k\rangle_{\mathrm{nG}}
+
\langle \hat a_j^\dagger \hat a_k^\dagger\rangle_{\mathrm{nG}},
\label{eq:xx_from_a_appendix}
\end{equation}
\begin{equation}
\langle \hat p_j \hat p_k\rangle_{\mathrm{nG}}
=
-\langle \hat a_j \hat a_k\rangle_{\mathrm{nG}}
+\langle \hat a_j \hat a_k^\dagger\rangle_{\mathrm{nG}}
+\langle \hat a_j^\dagger \hat a_k\rangle_{\mathrm{nG}}
-\langle \hat a_j^\dagger \hat a_k^\dagger\rangle_{\mathrm{nG}},
\label{eq:pp_from_a_appendix}
\end{equation}
\begin{equation}
\langle \hat x_j \hat p_k\rangle_{\mathrm{nG}}
=
\frac{
\langle \hat a_j \hat a_k\rangle_{\mathrm{nG}}
-\langle \hat a_j \hat a_k^\dagger\rangle_{\mathrm{nG}}
+\langle \hat a_j^\dagger \hat a_k\rangle_{\mathrm{nG}}
-\langle \hat a_j^\dagger \hat a_k^\dagger\rangle_{\mathrm{nG}}
}{i}.
\label{eq:xp_from_a_appendix}
\end{equation}
The covariance matrix of the de-Gaussified state is then obtained from the usual symmetrized definition
\begin{equation}
\sigma^{\mathrm{nG}}_{mn}
=
\frac{1}{2}
\langle \hat \xi_m \hat \xi_n + \hat \xi_n \hat \xi_m\rangle_{\mathrm{nG}}
-
\langle \hat \xi_m\rangle_{\mathrm{nG}}
\langle \hat \xi_n\rangle_{\mathrm{nG}}.
\end{equation}

In the main text, the readout is formed from the covariance matrix of the measured $x$ quadratures. In that case only the first moments $\langle \hat x_j\rangle_{\mathrm{nG}}$ and the second moments \eqref{eq:xx_from_a_appendix} are needed. The more general formulas above are included for completeness and to make clear that the same formalism immediately extends to first moments, full covariance matrices, and alternative quadrature choices.

The formalism above is exact for any finite product of creation and annihilation operators acting on a Gaussian state. In particular, it treats photon subtraction and photon addition symmetrically. This is sufficient for all simulations reported in the main text, where only a small number of such operations is considered. More complicated non-Gaussian circuits built from alternating Gaussian transformations and multiple non-Gaussian layers can be reduced to the same ingredients, at the expense of longer operator strings and therefore a larger Wick expansion.

\section{Finite-resource model}
\label{app:finite_resources}

In practice, the observables used for the readout are not known exactly, but must be estimated from a finite ensemble of homodyne outcomes. In this appendix we describe how finite-resource effects are incorporated into the numerical analysis.

Let $\hat{\bm{\xi}}_{\!x}$ denote the vector of measured quadratures used in the readout. In the main text this is the vector of selected $x$ quadratures. For a given input, repeated measurements produce a sample
\begin{equation}
\bm{\xi}^{(1)},\bm{\xi}^{(2)},\ldots,\bm{\xi}^{(M)}\in\mathbb{R}^{m},
\end{equation}
where $M$ is the ensemble size and $m$ is the number of measured quadratures. The corresponding sample mean is
\begin{equation}
\bar{\bm{\xi}}=\frac{1}{M}\sum_{r=1}^{M}\bm{\xi}^{(r)},
\end{equation}
and the sample covariance matrix is
\begin{equation}
\hat{\Sigma}
=
\frac{1}{M-1}
\sum_{r=1}^{M}
\left(\bm{\xi}^{(r)}-\bar{\bm{\xi}}\right)
\left(\bm{\xi}^{(r)}-\bar{\bm{\xi}}\right)^\top.
\label{eq:sample_cov_appendix}
\end{equation}
The entries of $\hat{\Sigma}$ replace the exact covariance-matrix elements in the noisy simulations.

If the measured quadratures are jointly Gaussian with true covariance matrix $\Sigma$, then the sample covariance matrix has an exact Wishart distribution,
\begin{equation}
(M-1)\hat{\Sigma}\sim W_m(\Sigma,M-1),
\label{eq:wishart_appendix}
\end{equation}
where $W_m(\Sigma,\nu)$ denotes the $m$-dimensional Wishart distribution with scale matrix $\Sigma$ and $\nu$ degrees of freedom. This gives an exact finite-resource model for the Gaussian reservoirs studied in the main text.

Operationally, this means that for each ideal covariance matrix $\Sigma$ associated with a given input, a noisy estimate $\hat{\Sigma}$ can be generated directly from the Wishart law \eqref{eq:wishart_appendix}. The resulting noisy covariance-matrix entries are then used as computational nodes in exactly the same way as in the noiseless analysis.

For the non-Gaussian reservoirs considered here, the quadrature statistics are no longer exactly multivariate Gaussian, so the Wishart law is not exact. Nevertheless, for large ensemble size $M$, the empirical moments become sharply concentrated around their exact values. In particular, under standard finite-moment assumptions the sample mean satisfies a multivariate central limit theorem, and the fluctuations of the sample covariance matrix become asymptotically Gaussian around the true covariance matrix.

In the present work we model finite-resource effects for the non-Gaussian states by using the same Wishart law as in the Gaussian case, with the exact covariance matrix of the de-Gaussified state inserted as the scale matrix. This approximation is exact in the Gaussian limit and becomes increasingly accurate as $M$ grows. It also has the practical advantage of preserving positive semidefiniteness of the sampled covariance matrix by construction.

Thus, for both Gaussian and non-Gaussian reservoirs, finite-resource effects are incorporated by replacing the exact covariance matrix $\Sigma$ associated with a given input by a sampled matrix $\hat{\Sigma}$ drawn according to Eq.~\eqref{eq:wishart_appendix}, with the understanding that this is exact in the Gaussian case and an asymptotically accurate approximation in the weakly non-Gaussian case.

Once noisy covariance matrices have been generated, the corresponding computational nodes are assembled exactly as in the noiseless case. The information-processing capacities are then computed from the noisy observables using the same training and testing protocol as for the ideal data. In this way, the only difference between the noiseless and noisy simulations is that exact covariance-matrix elements are replaced by finite-sample estimates.

All noisy capacities reported in the main text are obtained by averaging over independent realizations of both the reservoir parameters and the finite-resource sampling procedure. This makes it possible to assess whether the excess above the Gaussian bound survives under experimentally realistic conditions.

The finite-resource model adopted here is intended as an operational approximation for the regime of interest, namely weakly non-Gaussian multimode states probed with large but finite ensembles. It is not a universal exact sampling law for arbitrary non-Gaussian states. In particular, for strongly non-Gaussian states or very small ensemble sizes, higher-order cumulants beyond the covariance matrix may influence the sampling statistics in ways not captured by the Wishart approximation. The model should therefore be interpreted as a controlled large-$M$ approximation that is exact for Gaussian states and appropriate for the weakly non-Gaussian regime studied in this work.

\end{document}